\documentclass[twoside,british]{paper}
\usepackage[T1]{fontenc}
\usepackage[latin9]{inputenc}
\usepackage{bm}
\usepackage{amsthm}
\usepackage{amsmath}
\usepackage{amssymb}
\usepackage{graphicx}
\usepackage{esint}
\usepackage[authoryear]{natbib}

\makeatletter
\numberwithin{equation}{section}
\numberwithin{figure}{section}
\theoremstyle{plain}
\newtheorem{thm}{\protect\theoremname}
  \theoremstyle{plain}
  \newtheorem{cor}[thm]{\protect\corollaryname}
  \theoremstyle{plain}
  \newtheorem{prop}[thm]{\protect\propositionname}
  \theoremstyle{remark}
  \newtheorem{rem}[thm]{\protect\remarkname}
  \theoremstyle{plain}
  \newtheorem{lem}[thm]{\protect\lemmaname}

\DeclareMathOperator*{\argmax}{arg\,max}

\makeatother

\usepackage{babel}
  \providecommand{\corollaryname}{Corollary}
  \providecommand{\lemmaname}{Lemma}
  \providecommand{\propositionname}{Proposition}
  \providecommand{\remarkname}{Remark}
\providecommand{\theoremname}{Theorem}

\begin{document}
\global\long\def\ba{\bm{\alpha}}
 \global\long\def\bt{\bm{\theta}}
\global\long\def\bmu{\bm{\mu}}
 \global\long\def\bl{\bm{\lambda}}
 \global\long\def\be{\bm{\eta}}
 \global\long\def\E{\mathbb{E}}
 \global\long\def\by{\mathbf{y}}
 \global\long\def\ys{\mathbf{y}^{\star}}
 \global\long\def\bS{\bm{\Sigma}}
 \global\long\def\N{\mathcal{N}}
 \global\long\def\I{\mathbb{I}}
 \global\long\def\R{\mathbb{R}}
 \global\long\def\bb{\bm{\beta}}
\global\long\def\O{\mathcal{O}}
\global\long\def\Rc{\mathcal{R}}
\global\long\def\Vc{\mathcal{V}}
\global\long\def\bG{\bm{\Gamma}}
\global\long\def\bD{\bm{\Delta}}
\global\long\def\bn{\bm{\nu}}
\global\long\def\L{\mathcal{L}}
\global\long\def\Lt{\mathcal{L}\left(\bt\right)}
\global\long\def\M{\mathcal{M}}
\global\long\def\Mt{\mathcal{M}\left(\bt,\bn\right)}
\global\long\def\Mtn{\mathcal{M}\left(\bt,\nu\right)}
\global\long\def\bts{\bm{\theta}^{\star}}
\global\long\def\logitinv{\mathrm{logit}^{-1}}
\global\long\def\H{\mathcal{H}}
\global\long\def\bxi{\bm{\xi}}

\title{The Poisson transform for unnormalised statistical models}

\author{Simon Barthelmé, Nicolas Chopin}
\maketitle
\begin{abstract}
Contrary to standard statistical models, \emph{unnormalised }statistical
models only specify the likelihood function up to a constant. While
such models are natural and popular, the lack of normalisation makes
inference much more difficult. Extending classical results on the
multinomial-Poisson transform \citep{Baker:MultinomialPoissonTransform},
we show that inferring the parameters of a unnormalised model on a
space $\Omega$ can be mapped onto an equivalent problem of estimating
the intensity of a Poisson point process on $\Omega$. The unnormalised
statistical model now specifies an \emph{intensity function} that
does not need to be normalised. Effectively, the normalisation constant
may now be inferred as just another parameter, at no loss of information.
The result can be extended to cover non-IID models, which includes
for example unnormalised models for sequences of graphs (dynamical
graphs), or for sequences of binary vectors. As a consequence, we
prove that unnormalised parameteric inference in non-IID models can
be turned into a semi-parametric estimation problem. Moreover, we
show that the \emph{noise-contrastive estimation }method of \citet{GutmannHyvarinen:NoiseContrastEstUnStatMod}
can be understood as an approximation of the Poisson transform, and
extended to non-IID settings. We use our results to fit spatial Markov
chain models of eye movements, where the Poisson transform allows
us to turn a highly non-standard model into vanilla semi-parametric
logistic regression. 
\end{abstract}
Unnormalised statistical models are a core tool in modern machine
learning, especially deep learning \citep{salakhutdinov2009deep},
computer vision (\emph{Markov random fields}, \citealp{Wang:MRFModelingInferenceLearning})
and statistics for point processes \citep{gu2001maximum}, network
models \citep{caimo2011bayesian}, directional data \citep{walker2011posterior}.
They appear naturally whenever one can best describe data as having
to conform to certain features: we may then define an energy function
that measures how well the data conform to these constraints. While
this way of formulating statistical models is extremely general and
useful, immense technical difficulties may arise whenever the energy
function involves some unknown parameters which have to be estimated
from data. The reason is that the normalisation constant (which ensures
that the distribution integrates to one) is in most cases impossible
to compute. This prevents direct application of classical methods
of maximum likelihood or Bayesian inference, which all depend on the
unknown normalisation constant.

Many techniques have been developed in recent years for such problems,
including contrastive divergence \citep{Hinton:TrainingProductExpertsContrastDiv,BengioDelalleau:JustifyingGeneralizingContrDiv},
noise-contrastive estimation \citep{GutmannHyvarinen:NoiseContrastEstUnStatMod}
and various forms of MCMC for Bayesian inference \citep{Moller:EfficientMCMCforDistrIntractNormConstants,Murray:MCMCDoublyIntractableDistr,Girolami:PlayingRussianRoulette}.
The difficulty is compounded when unnormalised models are used for
non-IID data, either sequential data, or data that include covariates.
If the data form a sequence of length $n$, there are now $n$ normalisation
constants to approximate. In our application we look at models of
spatial Markov chains, where the transition density of the chain is
specified up to a normalisation constant, and again one normalisation
constant needs to be estimated per observation. 

In the first Section, we show that unnormalised estimation is tightly
related to the estimation of point process intensities, and formulate
a \emph{Poisson transform }that maps the log-likelihood of a model
$\Lt$ into an equivalent cost function $\Mt$ defined in an expanded
space, where the latent variables $\bn$ effectively estimate the
normalisation constants. In the case of non-IID unnormalised models
we show further that optimisation of $\Mt$ can be turned into a semi-parametric
problem and adressed using standard kernel methods. In the second
section, we show that the \emph{noise-contrastive divergenc}e\emph{
}described in of \citet{GutmannHyvarinen:NoiseContrastEstUnStatMod}
arises naturally as a tractable approximation of the Poisson transform,
and that this new interpretation lets us extend its use to non-IID
models. (\citet{GutmannHyvarinen:NoiseContrastEstUnStatMod} call
the technique ``noise-contrastive estimation'', but we use the term
noise-contrastive \emph{divergence }to designate the corresponding
cost function.) Finally, we apply these results to a class of unnormalised
spatial Markov chains that are natural descriptions of eye movement
sequences.

\section{Relationship to prior work }

Some of the ideas we use here have appeared under different forms
in classical statistics, machine learning and spatial statistics.
The Poisson transform generalises the multinomial-Poisson transform
developed by \citet{Baker:MultinomialPoissonTransform}. It is also
a special case of a general family of Bregman divergences introduced
by \citet{GutmannHirayama:BregmanDivergenceGeneralFramework}, a special
case of another family by \citet{Pihlaja:AFamilyCompEffSimpleEst},
and finally can also be viewed as an empirical version of the generalised
Kullback-Leibler divergence for unnormalised measures \citep{Minka:DivMeasuresMP}. 

Noise-contrastive learning is studied in \citet{GutmannHirayama:BregmanDivergenceGeneralFramework},
although the relationship between logistic regression and estimation
has been noted in other places (for example, in the spatial statistics
literature, see \citealp{Baddeley:SpatialLogisticRegAndChangeOfSupport},
\citealp{Baddeley:LogisticRegSpatialGibbsPointProc}). We go further
here in showing that the divergence defined by NCL converges uniformly
to the Poisson transform, giving it a new interpretation as an approximate
likelihood rather than just a divergence. 

\citet{MnihKavukcuoglu:LearningWordEmbeddings} and \citet{MnihTeh2012:FastSimpleAlgorithmTrainingNeuralProb}
use the NCL technique in a class of non-IID unnormalised models. However,
in the interest of computation time, they ignore normalisation constants.
The results given here indicate clearly that neglecting normalisation
constants leads in the general case to non-convergent estimators,
as illustrated in Section \ref{sub:A-toy-example}. Instead we develop
a semi-parametric framework for non-IID estimation, which is both
much faster than purely parametric techniques, as well as convergent.

\section{The Poisson transform \label{sec:The-Poisson-transform}}

In this section we show how unnormalised likelihoods can be turned
into Poisson process likelihoods at no loss of information. We call
the procedure the Poisson transform, as it generalises the Poisson-multinomial
transform \citep{Baker:MultinomialPoissonTransform}. We give two
interpretations, one in terms of upper-bound maximisation, and one
in terms of generalised KL divergences. We begin with the IID case,
with the generalisation to non-IID data treated further into the text.

\subsection{Background on Poisson point processes\label{sub:Background-on-PP}}

Poisson point processes are described at length in \citet{Kingman:PoissonProcesses},
and we only give here the merest outline. A \emph{Inhomogeneous Poisson
point process }(IPP) with intensity function $\lambda\left(\mathbf{y}\right)\geq0$
over space $\Omega$ defines a distribution over the set of countable
subsets $\mathcal{S}$ of $\Omega$, in such a way that, for any measurable
subset $\mathcal{A}\subseteq\Omega$,
\[
\#\left\{ \mathcal{S}\cap\mathcal{A}\right\} \sim\mbox{Poi}\left(\lambda_{\mathcal{A}}\right),\quad\lambda_{\mathcal{A}}=\int_{\mathcal{A}}\lambda\left(\mathbf{y}\right)\mbox{d}\mathbf{y},
\]
assuming $\lambda_{\mathcal{A}}<+\infty$. In words, the number of
points to be found in subset $\mathcal{A}$ has a Poisson distribution,
with expectation given by the integral of the intensity function within
$\mathcal{A}$; in discrete spaces the integral may of course be interpreted
as a sum. In particular, provided $\int\lambda\left(\mathbf{y}\right)\mbox{d}\mathbf{y}<+\infty$,
the cardinal $n$ of $\mathcal{S}$ is finite, and has a Poisson distribution
with expectation equal to the integral of $\lambda\left(\mathbf{y}\right)$
over the domain (the fact follows from taking $\mathcal{A}=\Omega$).
Assuming again $\int_{\Omega}\lambda\left(\mathbf{y}\right)\mbox{d}\mathbf{y}<+\infty$,
the log-likelihood of observing set $\mathcal{S}$ given the intensity
function $\mathbf{\lambda}$ is given by:
\begin{equation}
\log p\left(\mathcal{S}|\lambda\right)=\sum_{\mathbf{y}_{i}\in S}\log\lambda\left(\mathbf{y}_{i}\right)-\int_{\Omega}\lambda\left(\mathbf{y}\right)\mbox{d}\mathbf{y}.\label{eq:log-lik-Poisson}
\end{equation}

\subsection{The Poisson transform in the IID case\label{sub:Poisson-transform-IID}}

The Poisson transform is simply stated: when we have $n$ observations
from an unnormalised model on $\Omega$, we may treat them as the
realisation of a certain point process at no loss of information.
This results in a mapping from a likelihood function $\Lt$ to another,
which we note $\Mtn$, in an expanded space. $\Mtn$ has the same
global maximum as $\Lt$ and confidence intervals are preserved. 

First, the log-likelihood function for $n$ IID observations $\mathbf{y}_{i}$
from an unnormalised model $p(\mathbf{y}|\bt)\propto\exp\left\{ f_{\bt}(\mathbf{y})\right\} $
can be written as:
\begin{equation}
\mathcal{L}(\bt)=\sum_{i=1}^{n}f_{\bt}(\mathbf{y}_{i})-n\log\left(\int_{\Omega}\mbox{exp}\left\{ f_{\bt}\left(\mathbf{y}\right)\right\} \mbox{d}\mathbf{y}\right)\label{eq:standard-log-likelihood}
\end{equation}
and the ML estimate of $\bt$ is the maximum of $\Lt$. We introduce
the following alternative likelihood function: 
\begin{equation}
\Mtn=\sum_{i=1}^{n}\left\{ f_{\bt}(\mathbf{y}_{i})+\nu\right\} -n\int_{\Omega}\mbox{exp}\left\{ f_{\bt}(\mathbf{y})+\nu\right\} \mbox{d}\mathbf{y}\label{eq:poisson-transformed-loglik}
\end{equation}
which by \eqref{eq:log-lik-Poisson} is, up to additive constant $n\log(n)$,
the IPP likelihood on $\Omega$ for intensity function 
\[
\lambda\left(\mathbf{y}\right)=\mbox{exp}\left\{ f_{\bt}(\mathbf{y})+\nu+\log(n)\right\} .
\]
Our first theorem shows that maximum likelihood estimation of $\bt$
via $\Lt$ or via $\Mtn$ is equivalent. 
\begin{thm}
The set of points $\bts$ such that $\bts\in\underset{\bt\in\Theta}{\argmax}\,\Lt$
matches the set of points $\tilde{\bt}$ such that $(\tilde{\bt},\tilde{\nu})\in\underset{\bt\in\Theta,\nu\in\R}{\argmax}\,\Mtn$
for some $\tilde{\nu}$. In particular, if $\underset{\bt\in\Theta}{\argmax}\,\Lt$
is a singleton, then so is $\underset{\bt\in\Theta,\nu\in\R}{\argmax}\,\Mtn$.\end{thm}
\begin{proof}
For a fixed $\bt$, $\Mtn$ admits a unique maximum in $\nu$ at $\nu^{\star}(\bt)=-\log\int_{\Omega}\mbox{exp}\left\{ f_{\bt}(\mathbf{y})\right\} \mbox{d}\mathbf{y}$,
hence $\Mtn\leq\mathcal{M}(\bt,\nu^{\star}(\bt))=\Lt-n$.
\end{proof}
There are several remarks to make at this stage. First, since $\nu^{\star}(\bt)=-\log\int_{\Omega}\mbox{exp}\left\{ f_{\bt}(\mathbf{y})\right\} \mbox{d}\mathbf{y}$,
maximising $\mathcal{M}(\bt,\nu)$ can be interpreted as estimating
the normalisation constant along with the parameters. There is no
estimation cost incurred in treating the normalisation constant as
a free parameter, since the global maxima of $\Lt$ and $\Mtn$ are
the same.

Second, the usual way of computing confidence intervals for $\bt$
is to invert the Hessian of $\Lt$ at the mode. We show in the Appendix
that the same confidence intervals can be obtained from the Hessian
of $\Mtn$ at the mode, so that the Poisson transform does not introduce
any over or under-confidence. In addition, the Poisson-transformed
likelihood can be used for \emph{penalised }likelihood maximisation
(see Application), does not introduce any spurious maxima, and in
exponential families it can even be shown to preserve concavity (see
Appendix). 

Third, at this point we do not yet have a practical way of computing
$\mathcal{M}(\bt,\nu)$, since we have assumed that integrals of the
form $\int_{\Omega}\mbox{exp}\left\{ f_{\bt}(\mathbf{y})+\nu\right\} \mbox{d}\mathbf{y}$
are intractable. The problem of approximating $\mathcal{M}(\bt,\nu)$
is dealt with in Section \ref{sec:Practical-approximations}, where
we will see that among other possibilities it can be approximated
by logistic regression via noise-contrastive divergence. 

Before we deal with practical ways of approximating $\mathcal{M}(\bt,\nu)$,
we first generalise the Poisson transform to non-IID data.

\subsection{The Poisson transform in the non-IID case\label{sub:Poisson-transform-non-IID}}

In the non-IID case we still have $n$ datapoints $\mathbf{y}_{1}\ldots\mathbf{y}_{n}\in\Omega^{n}$
but their distribution is allowed to vary. For example the $n$ datapoints
might form a Markov chain with (unnormalised) transition density
\[
p_{\bt}(\mathbf{y}_{t}|\mathbf{y}_{t-1})\propto\exp\left\{ f_{\bt}(\mathbf{y}_{t}|\mathbf{y}_{t-1})\right\} 
\]
which leads to the log-likelihood
\begin{equation}
\mathcal{L}(\bt)=\sum_{t=1}^{n}\left[f_{\bt}(\mathbf{y}_{t}|\mathbf{y}_{t-1})-\log\int_{\Omega}\mbox{exp}\left\{ f_{\bt}\left(\mathbf{y}|\mathbf{y}_{t-1}\right)\right\} \,\mbox{d}\mathbf{y}\right].\label{eq:loglik-sequential-case}
\end{equation}
(The initial point $\mathbf{y}_{0}$ is treated as a constant.) Another
example is models with covariates $\mathbf{x}_{i}$, expressed as
$p(\mathbf{y}_{i}|\mathbf{x}_{i},\bt)\propto\exp\left\{ f_{\bt}(\mathbf{y}_{i}|\mathbf{x}_{i})\right\} $.
These two cases are highly similar and for brevity we focus on the
sequential case, which we use in our application. 

Our first step is to extend the Poisson transform \eqref{eq:poisson-transformed-loglik}
to yield a function $\Mt$ where $\bn$ is now a vector of dimension
$n$ (one per conditional distribution), $\bm{\nu}=(\nu_{1},\ldots,\nu_{n})$
and 
\begin{align}
\Mt & =\sum_{t=1}^{n}\left\{ f_{\bt}(\mathbf{y}_{t}|\mathbf{y}_{t-1})+\nu_{t-1}\right\} \nonumber \\
 & -\int_{\Omega}\left[\sum_{t=1}^{n}\mbox{exp}\left\{ f_{\bt}\left(\mathbf{y}|\mathbf{y}_{t-1}\right)+\nu_{t-1}\right\} \right]\,\mbox{d}\mathbf{y}.\label{eq:poisson-transform-sequential}
\end{align}

\begin{thm}
The set of points $\bts$ such that $\bts\in\underset{\bt\in\Theta}{\argmax}\,\Lt$
matches the set of points $\tilde{\bt}$ such that $\left(\tilde{\bt},\bn^{\star}\right)=\underset{\bt\in\Theta,\bn\in\R^{n}}{\argmax}\,\Mt$.\label{thm:2}\end{thm}
\begin{proof}
The proof is along the same lines as that of the Theorem 1: maximising
$\Mt$ in $\nu_{t-1}$ gives $\nu_{t-1}^{\star}(\bt)=-\log\int_{\Omega}\exp\left\{ f_{\bt}(\by|\by_{t-1})\,\mbox{d}\mathbf{y}\right\} $,
and $\mathcal{M}(\bt,\bm{\nu}^{\star}(\bt))=\mathcal{L}(\bt)-n$.
\end{proof}
Note that while $\Lt$ involves the sum of $n$ separate integrals,
$\Mt$ involves a single integral over a sum.  Further, since
\[
\nu_{t-1}^{\star}\left(\bt\right)=-\log\left(\int_{\Omega}\mbox{exp}\left\{ f_{\bt}\left(\mathbf{y}|\mathbf{y}_{t-1}\right)\right\} \mbox{d}\mathbf{y}\right)
\]
the optimal value of $\nu_{t-1}$ is a function of \textbf{$\mathbf{y}_{t-1}$
}\emph{only}. This means that we can think of the integration constants
as (hopefully smooth) functions of the previous point $\mathbf{y}_{t-1}$.
This leads to the following result: let $\mathcal{F}$ denote an appropriate
function space that contains the function $\chi:\Omega\rightarrow\mathbb{R}$
such that $\chi(\bm{u})=-\log\int_{\Omega}\mbox{exp}\left\{ f_{\bt}\left(\mathbf{y}|\bm{u}\right)\right\} \mbox{d}\mathbf{y}$.
We introduce the following functional
\begin{align}
\M_{\chi}\left(\bt,\chi\right) & =\sum_{t=1}^{n}\left\{ f_{\bt}(\mathbf{y}_{t}|\mathbf{y}_{t-1})+\chi(\mathbf{y}_{t-1})\right\} \nonumber \\
 & -\int_{\Omega}\sum_{t}\mbox{exp}\left\{ f_{\bt}\left(\mathbf{y}|\mathbf{y}_{t-1}\right)+\chi(\mathbf{y}_{t-1})\right\} \mbox{d}\mathbf{y}.\label{eq:poisson-transform-seq-nonpar}
\end{align}

\begin{cor}
The set of points $\bts$ such that $\bts\in\underset{\bt\in\Theta}{\argmax}\,\Lt$
matches the set of points $\tilde{\bt}$ such that $\left(\tilde{\bt},\chi\right)\in\underset{\bt\in\Theta,\chi\in\mathcal{F}}{\mbox{argmax}}\,\M_{\chi}\left(\bt,\chi\right)$.
\label{thm:semi-parametricPoissonmodel}
\end{cor}
We can use this Corollary to turn inference on unnormalised models
into a semiparametric problem, where $\bt$ is estimated parametrically
and the normalisation constants are estimated as a non-parametric
function $\chi(\mathbf{y}_{t-1})$. 

In the formulation used by Corollary \ref{thm:semi-parametricPoissonmodel}
there exists possibly (uncountably) many optimal normalisation functions
$\chi$, ie. functions that solve $\underset{\bt\in\Theta,\chi\in\mathcal{F}}{\mbox{argmax}}\,\M_{\chi}\left(\bt,\chi\right)$.
All that is required is that they interpolate the values of the normalisations
constants for the various $\mathbf{y}_{t-1}$ in the dataset. To get
a unique optimal normalisation function we need to regularise the
non-parametric part. 

A classical way to solve non-parametric regression problems is to
model the non-parametric part as belonging to a Reproducible Kernel
Hilbert Space (RKHS), and to add regularisation by including a penalty.
The following result shows that penalised non-parametric estimation
can be made consistent, and the optimal normalisation function becomes
uniquely defined. 
\begin{prop}
Let $\H$ denote a RKHS, with kernel function $k(\mathbf{y},\mathbf{y}')$
and $\left|f\right|_{_{\H}}$ the corresponding norm. Suppose $\H$
contains one optimal normalisation function, i.e. there exists an
$\chi^{*}(\mathbf{u})\in\H$, with $\left|\chi^{\star}\right|_{\H}<\infty$
Then there exists a value $\lambda_{0}>0$ such that the set of maximum
likelihood points $\bts\in\underset{\bt\in\Theta}{\argmax}\,\Lt$
matches the set of points penalised estimates $\tilde{\bt}$ defined
by: 
\begin{equation}
\left(\tilde{\bt},\chi\right)\in\underset{\bt\in\Theta,\chi\in\mathcal{\H}}{\mbox{argmax}}\,\M_{\chi}\left(\bt,\chi\right)-\lambda\left|\chi\right|_{\H}^{2}\label{eq:penalised-semiparametric-problem}
\end{equation}

i.e., the penalised non-parametric Poisson estimator is equivalent
to the maximum-likelihood estimator. \end{prop}
\begin{proof}
The penalised problem is equivalent to the following constrained optimisation
problem:

\begin{eqnarray*}
\underset{\bt\in\Theta,\chi\in\mathcal{\H}}{\mbox{argmax}} & \, & \M_{\chi}\left(\bt,\chi\right)\\
\mbox{subject to} &  & \left|\chi\right|_{\H}^{2}\leq\rho
\end{eqnarray*}

for some value $\rho$ dependent on $\lambda$ (this follows from
writing the Lagrangian). By the assumption that there exists an optimal
normalisation function in $\H$ with finite norm, there exists a $\rho_{0}<\infty$
such that the constraint is irrelevant and solving the constrained
problen above is equivalent to solving the non-penalised problem $\underset{\bt\in\Theta,\chi\in\mathcal{F}}{\mbox{argmax}}\,\M_{\chi}\left(\bt,\chi\right)$
from Lemma \ref{thm:semi-parametricPoissonmodel}. Correspondingly
there exists a penalisation parameter $\lambda_{0}>0$ such that the
penalised estimate (\ref{eq:penalised-semiparametric-problem}) matches
the non-penalised estimate. \end{proof}
\begin{rem}
For fixed $\bt$, $\underset{\chi\in\mathcal{\H}}{\mbox{argmax}}\,\M_{\chi}\left(\bt,\chi\right)-\lambda\left|\chi\right|_{\H}^{2}$
has a unique solution that can be expressed as $\chi\left(\mathbf{u}\right)=\sum\alpha_{t-1}k(\mathbf{u},\mathbf{y}_{t-1})$\end{rem}
\begin{proof}
The result follows from a straightforward application of the Representer
Theorem (see \citealp{ScholkopfSmola:LWK}, page 90). 
\end{proof}
We have only established so far that there exists a value $\lambda_{0}$
so that the penalised non-parametric estimator is equivalent to the
ML estimator. We cannot expect to know that value in advance, and
so $\lambda_{0}$ needs to be estimated from the data. The following
Corrolary comes to the rescue:
\begin{cor}
Note $\bt\left(\lambda\right),\chi\left(\lambda\right)$ the solution
for the penalised problem (eq. \eqref{eq:penalised-semiparametric-problem})
with regularisation parameter $\lambda$. For all $\lambda\leq\lambda_{0}$,
$\M_{\chi}\left(\bt\left(\lambda\right),\chi\left(\lambda\right)\right)=\M_{\chi}\left(\bt\left(\lambda_{0}\right),\chi\left(\lambda_{0}\right)\right)$,
i.e. there is no further improvement to the optimal value of the Poisson
transform by relaxing the penalty beyond $\lambda_{0}$.\label{cor:Cost-doesnt-improve-beyond-lambda0}\end{cor}
\begin{proof}
The proof follows again from the constrained formulation. By $\lambda_{0}$
we have already found the optimal solution and there is no point relaxing
the constraint further. 
\end{proof}
What the result suggests is that we could start with a high value
for $\lambda$, perform the optimisation, and reduce the value of
$\lambda$ until the value of $\M_{\chi}\left(\bt\left(\lambda\right),\chi\left(\lambda\right)\right)$
stops improving. We will then have found the most ``simple'' function
that interpolates the normalisation constants. Unfortunately Corrolary
\ref{cor:Cost-doesnt-improve-beyond-lambda0} does not hold for noise-contrastive
divergence, and so a different strategy (such as cross-validation)
has to be used for selecting $\lambda$. We return to the issue in
the examples.

\section{Practical approximations for the Poisson transform\label{sec:Practical-approximations}}

The Poisson transform gives us an alternative likelihood function
for estimation, but one that still involves an intractable integral.
In this section we briefly describe some practical approximations.
One is based on importance sampling and leads to an unbiased estimate
of the gradient (meaning that novel stochastic gradient and approximate
Langevin sampling methods are possible). The second is based on logistic
regression: we show that the noise-contrastive divergence of \citet{GutmannHyvarinen:NoiseContrastEstUnStatMod}
approximates the Poisson-transformed likelihood. Using that connection,
estimation in any non-IID setting can be turned into a semiparametric
classification problem.

\subsection{Unbiased estimation of the gradient \label{sub:Unbiased-estimation-of-gradient}}

The first derivatives of $\Mtn$ (eq. \ref{eq:poisson-transformed-loglik})
equal:
\begin{eqnarray*}
\frac{1}{n}\frac{\partial}{\partial\bt}\Mtn & = & \frac{1}{n}\sum_{i=1}^{n}\frac{\partial}{\partial\bt}f_{\bt}(\mathbf{y}_{i})-\int_{\Omega}\frac{\partial}{\partial\bt}f_{\bt}(\mathbf{y}_{i})\mbox{exp}\left\{ f_{\bt}(\mathbf{y})+\nu\right\} \mbox{d}\mathbf{y}\\
\frac{1}{n}\frac{\partial}{\partial\nu}\Mtn & = & 1-\int_{\Omega}\mbox{exp}\left\{ f_{\bt}(\mathbf{y})+\nu\right\} \mbox{d}\mathbf{y}
\end{eqnarray*}

The integrals on the right hand side can be estimated unbiasedly by
Monte Carlo, which is not true in general for the untransformed likelihood.
The availability of an unbiased estimator for the gradient means that
stochastic gradient algorithms (and their MCMC counterpart, approximate
Langevin sampling, \citealp{Welling:BayesianLearningViaStochGradLangevin})
can be applied directly. The resulting method has a straightforward
interpretation, since we simply adjust $\nu$ until $\mbox{exp}\left\{ f_{\bt}(\mathbf{y})+\nu\right\} $
normalises to 1 on average.

\subsection{Logistic likelihood as an approximation: IID case\label{sub:Logistic-likelihoods}}

In this section we show how to approximate Poisson-transformed likelihoods,
see \eqref{eq:poisson-transformed-loglik} and \eqref{eq:loglik-sequential-case},
using logistic regression. Reductions to logistic regression appear
in many places in the statistical literature. In the context of estimation
it is described in the well-known textbook of \citet{Hastie:ESL}
and in detail in \citet{Baddeley:SpatialLogisticRegAndChangeOfSupport}.
The use of logistic regression to estimate \emph{normalisation constants
}is described in \citet{Geyer:EstimatingNormConstantsReweightMixtures}.
Recently \citet{GutmannHyvarinen:NoiseContrastEstUnStatMod} introduced
a more general theory which they call ``noise-contrastive divergence'',
and show that logistic regression can be used for joint estimation
of parameters and normalisation constants. 

The essence of noise-contrastive divergence is to try and teach a
logistic classifier to tell true data $\mathcal{S}=\left\{ \by_{1},\ldots,\by_{n}\right\} $,
generated from $p_{\bt}(\by)$, from random reference data $\mathcal{R}=\left\{ \bm{r}_{1},\ldots,\bm{r}_{m}\right\} $,
generated from some distribution with density $q(\bm{r})$. Picking
a point $\bm{u}$ at random from $\mathcal{S}\cup\mathcal{R}$, and
denoting $z=1$ (resp. $z=0$) the event that $\bm{u}$ comes from
$\mathcal{S}$ (resp. $\mathcal{R}$), one obtains the following log
odds ratio: 
\begin{equation}
\log\frac{p\left(z=1\vert\bm{u}\right)}{p\left(z=0\vert\bm{u}\right)}=\log p_{\bt}\left(\bm{u}\right)-\log q(\bm{u})+\log\left(n/m\right).\label{eq:logistic-ratio-static}
\end{equation}

If we assume additionally that $p_{\bt}(\by)$ is unnormalised, $p_{\bt}(\by)\propto\exp\left\{ f_{\bt}(\by)\right\} $,
one may replace above, in the same spirit as in our Poisson transform,
the term $\log p_{\bt}(\bm{u})$ by $f_{\bt}(\bm{u})+\nu$, leading
to 

\begin{equation}
\log\frac{p\left(z=1\vert u\right)}{p\left(z=0\vert u\right)}=f_{\bt}(\bm{u})+\nu-\log q(\bm{u})+\log(n/m).\label{eq:logistic-ratio-unnormalised}
\end{equation}

This leads to following simple recipe: generate reference data $\mathcal{R}$,
then estimate jointly $(\bt,\nu)$ by fitting the logistic regression
\eqref{eq:logistic-ratio-unnormalised} to the dataset $\mathcal{S}\cup\mathcal{R}$,
with points in $\mathcal{S}$ (resp. $\mathcal{R}$) labelled as $z_{i}=1$
(resp. $z_{i}=0$). 

The obvious connection between our Poisson transform and the noise-contrastive
approach is that in both cases the log normalising constant is treated
as a free parameter. The following result reveals that this connection
is actually deeper. 
\begin{thm}
\label{thm:logistic-reg-tends-to-IPP}For fixed $\bt$, $\nu$, and
$\mathcal{S}=\left\{ \by_{1},\ldots,\by_{n}\right\} $, and under
the assumption that $f_{\bt}(\bm{y})-\log q(\bm{y})\leq C(\bt)$ for
all $\by\in\Omega$, the log-likelihood of the logistic regression
defined above: 
\begin{eqnarray*}
\mathcal{R}^{m}(\bt,\nu) & = & \sum_{i=1}^{n}\log\left[\frac{n\exp\left\{ f_{\bt}(\by_{i})+\nu\right\} }{n\exp\left\{ f_{\bt}(\by_{i})+\nu\right\} +mq(\by_{i})}\right]\\
 &  & +\sum_{j=1}^{m}\log\left[\frac{mq(\bm{r}_{j})}{n\exp\left\{ f_{\bt}(\bm{r}_{j})+\nu\right\} +mq(\bm{r}_{j})}\right]
\end{eqnarray*}
is such that 
\begin{equation}
\mathcal{R}^{m}(\bt,\nu)+n\log(m/n)+\sum_{i=1}^{n}\log q(\bm{y}_{i})\rightarrow\mathcal{M}(\bt,\nu)\label{eq:logistic-converges-to-IPP}
\end{equation}
 almost surely as $m\rightarrow+\infty$, relative to the randomness
induced by the reference points $\mathcal{R}=\left\{ \bm{r}_{1},\ldots,\bm{r}_{m}\right\} $.\end{thm}
\begin{proof}
See Appendix.
\end{proof}
The theorem above establishes that $\mathcal{R}^{m}(\bt,\nu)$ converges
to $\mathcal{M}(\bt,\nu)$ pointwise (up to a constant). Uniform convergence
(with respect to $\bt$) may be proved under stronger conditions.
As a corollary, one obtains that the MLE based on $\mathcal{R}^{m}(\bt,\nu)$
converges to the intractable MLE of $\mathcal{M}(\bt,\nu)$ as $m\rightarrow+\infty$. 
\begin{thm}
\label{thm:MLEconverges}Assume that (i) $\Theta$ is a bounded set,
that (ii) $\left|f_{\bt}(\by)-\log q(\by)\right|\leq C$ for some
$C>0$ and all $\by\in\Omega$, that (iii) $\left|f_{\bt}(\by)-f_{\bt'}(\by)\right|\leq\kappa(\by)\left\Vert \bt-\bt'\right\Vert $
for all $\by\in\Omega$ and $\bt,\bt'\in\Theta$, with $\E_{q}[\kappa]<\infty$,
that (iv) there exists $\hat{\bt}$ such that $\L\left(\hat{\bt}\right)>\sup_{d(\hat{\bt},\bt)\geq\epsilon}\L\left(\bt\right)$,
for any $\epsilon>0$. Then for fixed $\mathcal{S}=\{\by_{1},\ldots,\by_{n}\}$,
and $\left(\tilde{\bt}_{m},\tilde{\nu}_{m}\right)$ such that $\mathcal{R}^{m}(\tilde{\bt}_{m},\tilde{\nu}_{m})=\sup_{\left(\bt,\nu\right)\in\Theta\times\R}\mathcal{R}^{m}(\bt,\nu)$,
one has 
\[
\tilde{\bt}_{m}\rightarrow\hat{\bt}\quad\mbox{a.s. }
\]
as $m\rightarrow+\infty$, relative to the randomness induced by the
reference points $\mathcal{R}=\{\bm{r}_{1},\ldots,\bm{r}_{m}\}$.\end{thm}
\begin{proof}
See Appendix. 
\end{proof}
In particular, the limit of $\tilde{\bt}_{m}$ as $m\rightarrow+\infty$
has the same properties as the MLE of $\mathcal{L}(\bt)$, and thus
is consistent, and asymptotically efficient. The theorem above assumes
implicitly that the MLE of the logistic regression (with log-likelihood
$\mathcal{R}^{m}(\bt,\nu)$) is well defined, but this is a mild assumption:
e.g. if the considered model corresponds to an exponential family,
$f_{\bt}(\by)=\bt^{T}S(\by)$, then it is easy to check that $\mathcal{R}^{m}(\bt,\nu)$
is a concave function of $\left(\bt,\nu\right)$.

\subsection{Logistic likelihood as an approximation: non IID case}

Putting together Theorem 4 and the results in Section \ref{sub:Poisson-transform-non-IID}
leads to the following extension of noise-contrastive divergence to
non-IID problems. For an unnormalised Markov model $p_{\bt}(\by_{t}|\by_{t-1})\propto\exp\left\{ f_{\bt}(\by_{t}|\by_{t-1})\right\} $,
for data $\mathcal{S}=\left\{ \by_{1},\ldots,\by_{n}\right\} $, generate
$m=kn$ reference datapoints $\bm{r}_{jt}$ from kernel $q(\bm{r}_{jt}|\by_{t-1})$,
$j=1,\ldots,k$ (i.e. $k$ points $\bm{r}_{j}$ are generated from
ancestor $\by_{t-1}$, for each $t$), then fit the semi-parametric
logistic regresssion model that corresponds to the log odds ratio
function:
\begin{equation}
\log\frac{p\left(z=1\vert\bm{u}_{t-1},\bm{u}_{t}\right)}{p\left(z=0\vert\bm{u}_{t-1},\bm{u}_{t}\right)}=f_{\bt}(\bm{u}_{t}|\bm{u}_{t-1})+\chi(\bm{u}_{t-1})-\log q(\bm{u}_{t}|\bm{u}_{t-1})+\log(n/m)\label{eq:log-odds-seq}
\end{equation}
where $(\bm{u}_{t-1},\bm{u}_{t})$ represents a pair taken at random
from $\{(\by_{t-1},\by_{t})\}\cup\left\{ (\bm{y}_{t-1},\bm{r}_{jt})\right\} $.
The parameters of this logistic model are vector $\bt$, scalar $\nu$,
and function $\chi:\mathcal{Y}\rightarrow\mathbb{R}$, which is why
this model is indeed semi-parametric. In practice, fitting such a
model is easily achieved using an appropriate regulariser (we use
smoothing splines in our application). 

The interpretation of the above procedure follows the same lines as
in the previous section: for $m\rightarrow+\infty$, the log-likelihood
of this logistic model converges to that of the semi-parametric Poisson
model defined in Theorem \ref{thm:semi-parametricPoissonmodel}; in
particular, $\chi$ must be seen as an estimator of the (typically
smooth) function $\bm{y}_{t-1}\rightarrow-\log\int\exp\left\{ f_{\bt}(\by|\by_{t-1})\right\} \, d\by$. 

More generally, one may extend this approach to other non-IID models.
For instance, if $p_{\bt}(\by_{t})\propto\exp\left\{ f_{\bt}(\by_{t}|\bm{x}_{t})\right\} $,
where $\bm{x}_{t}$ are covariates, then fit the same type of semi-parametric
logistic regression as above, but with $\chi$ a function of covariates
$\bm{x}_{t}$.

\section{Applications to spatial Markov chains\label{sec:Application}}

\subsection{A toy example\label{sub:A-toy-example}}

We begin with a toy example that shows how inference based on the
Poisson Transform can be implemented in the non-IID case, and show
that semi-parametric inference using non-contrastive divergence can
be almost as efficient as maximum-likelihood (and much more efficient
than completely parametric non-contrastive divergence). In addition,
we will see that ignoring normalisation constants as done by \citet{MnihTeh:FastAndSimpleAlg}
and \citet{MnihKavukcuoglu:LearningWordEmbeddings} can lead to severe
bias. We have made available a detailed companion document for this
section, which includes all the code necessary to replicate our results
in R. 

Our toy example is a Markov chain in $[-1,1]$, with transition probability:
\begin{equation}
p_{\bt}(y_{t}\vert y_{t-1})\propto\exp\left\{ \theta_{1}y_{t}-\frac{1}{2}\theta_{2}\ (y_{t}-y_{t-1})^{2}\right\} \mathbb{I}_{[-1,1]}(y_{t}).\label{eq:toy-model}
\end{equation}

We picked this example because it is a simplified version of the spatial
Markov chains we study in the following section. Two realisations
from the chain are shown on Fig. \ref{fig:Two-realisations-from-toy}.

\begin{figure}
\begin{centering}
\includegraphics[width=8cm]{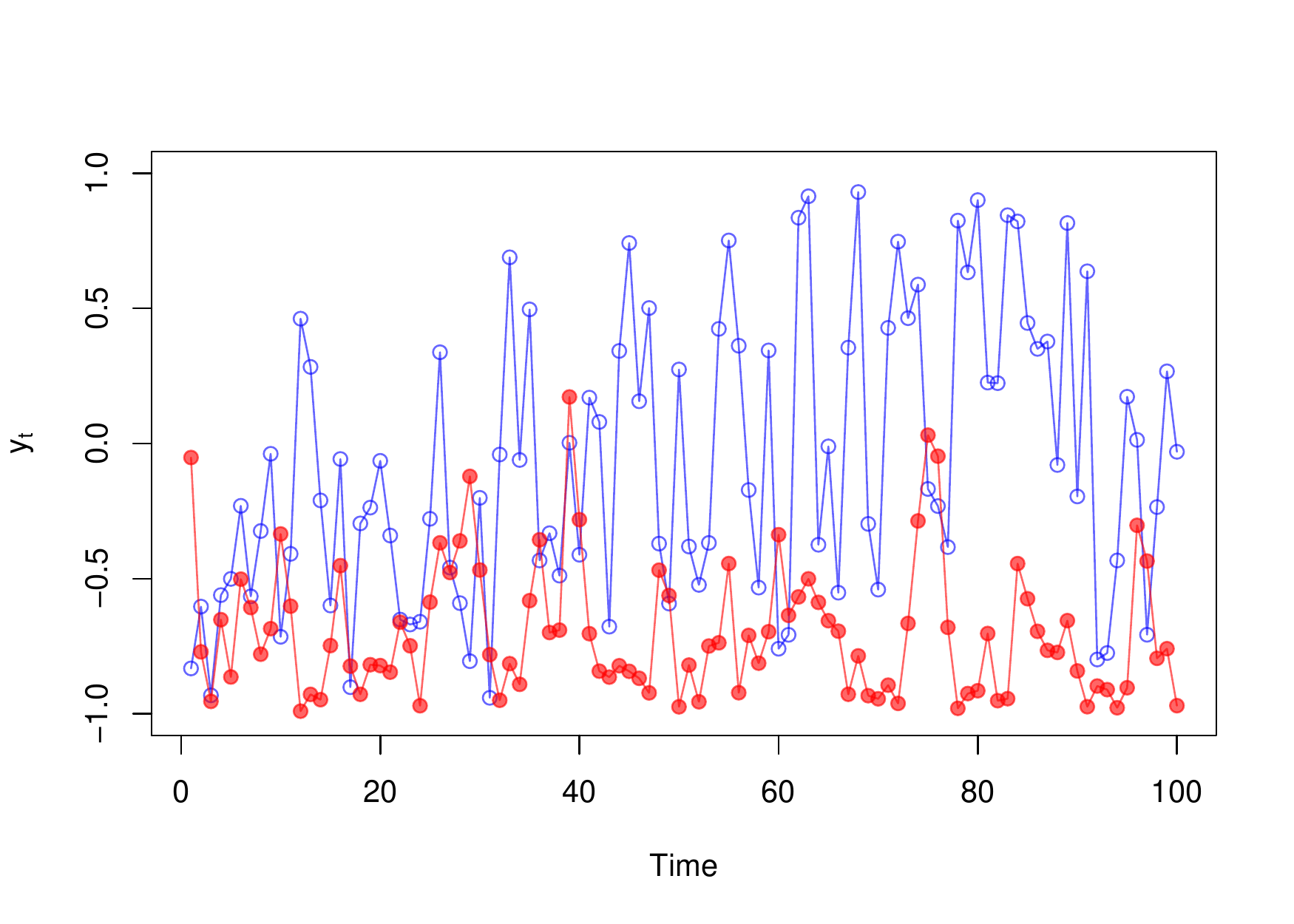}
\par\end{centering}

\protect\caption{Two realisations from the toy model, a Markov chain constrained to
the interval $[-1,1]$ (eq. \ref{eq:toy-model}). Red, unfilled dots:
$\theta_{1}=0,\,\theta_{2}=2$. Blue, solid dots: $\theta_{1}=-2,\,\theta_{2}=10$.
As expected, the second chain shows a bias towards negative values
as well as stronger autocorrelation.\label{fig:Two-realisations-from-toy}}
\end{figure}

In this one-dimensional example it is of course easy to compute the
normalisation constant using numerical integration, and thus maximum
likelihood inference is possible. To use non-constrastive divergence,
we need to pick a reference kernel, and here a uniform, IID distribution
does the job quite well: $q(y\vert y_{t-1})=\frac{1}{2}\mathbb{I}_{[-1,1]}(y).$

Positive examples for the logistic regression are formed from actual
pairs $\left(y_{t},y_{t-1}\right)$, negative examples are formed
from pairs $\left(r_{it},y_{t-1}\right)$, $i=1,\ldots,k$, i.e. one
replaces the actual value of $y_{t}$ with k uniform variates. Thus,
there are $k$ reference points per datapoint: $m=kn$.

We note $\left(u_{t},u_{t-1}\right)$ a generic point (either true
data, or reference data). The log-odds for the semi-parametric logistic
regression are then, injecting \eqref{eq:toy-model} into \eqref{eq:log-odds-seq}:
\begin{align}
m_{\bt}\left(u_{t}\right) & =\theta_{1}u_{t}-\frac{1}{2}\theta_{2}(u_{t}-u_{t-1})^{2}+\chi(u_{t-1})+\log(n/m)-\log\left(1/2\right),\nonumber \\
 & =\theta_{1}u_{t}+\theta_{2}d_{t}+\chi(u_{t-1})+\mbox{cst}.\label{eq:log-odds-toy}
\end{align}
where $d_{t}=\frac{1}{2}(u_{t}-u_{t-1})^{2}$. From a practical perspective,
the logistic regression can be performed with $u_{t}$ and $d_{t}$
entering as linear effects, and $\chi(u_{t-1})$ as a smooth, nonlinear
effect. The constant term may be added as an offset for completeness.
(It makes no practical difference since it can be absorbed into $\chi\left(u_{t-1}\right)$
or the intercept. One needs to include it only if intercepts are penalised.). 

The completely parametric variant of \eqref{eq:log-odds-toy} corresponds
to having a different intercept for every value of $u_{t-1}$. Alternatively,
neglecting the normalisation constants means replacing $\chi(u_{t-1})$
with an intercept term (or, put another way, forcing $\chi\left(u_{t-1}\right)$
to be constant). Semiparametric inference can be performed using R
package \textsf{mgcv} \citep{Wood:GAMIntroductionWithR}. 

To measure the efficiency of the various estimation methods, we simulated
realisations of the chain at a fixed parameter setting of $\theta_{1}=-2,\theta_{2}=50$
for increasing $n$. We also used two different values of $k$ (the
ratio of reference points to real data), $k=10$ and $k=30$. On each
simulation we picked two parameter values at random: $\theta_{1}\sim\mathcal{U}(-1,1),\theta_{2}\sim\mathcal{U}(\frac{1}{10},10)$,
generated $n$ datapoints, and obtained the 4 different estimates.
We used 300 repetitions for each value of $n$ and $k$. Results are
shown on Fig. \ref{fig:Results-estimation-toy}. 

Semiparametric inference performs almost as well as ML. Fully parametric
inference is much more variable, although it becomes better for larger
values of $k$. Indeed, theory predicts that it for large enough $k$
it becomes equivalent to ML. The variant of non-contrastive divergence
which neglects the normalisation constants performs quite well for
$\theta_{2}$ but shows asymptotic bias in $\theta_{1}$. The bias
comes from the missing non-linear effect $\chi\left(u_{t-1}\right)$,
which is projected on the linear effect for $u_{t}$. This happens
because the two are correlated through the dependencies in the chain.
Neglecting the normalisation constants then effectively leads to confounding. 

Contrary to the ideal Poisson transform (see correlary \ref{cor:Cost-doesnt-improve-beyond-lambda0}),
the non-contrastive divergence approximation is noisy and it is possible
to overfit the nonparametric term $\chi\left(u_{t-1}\right)$. Cross-validation
is a valid way of selecting the penalisation level, and here in practice
related criteria such as Generalised Cross-Validation and REML work
just as well. The results in Fig. \ref{fig:Results-estimation-toy}
are obtained using the default criterion (GCV). 

\begin{figure}
\begin{centering}
\includegraphics[width=14cm]{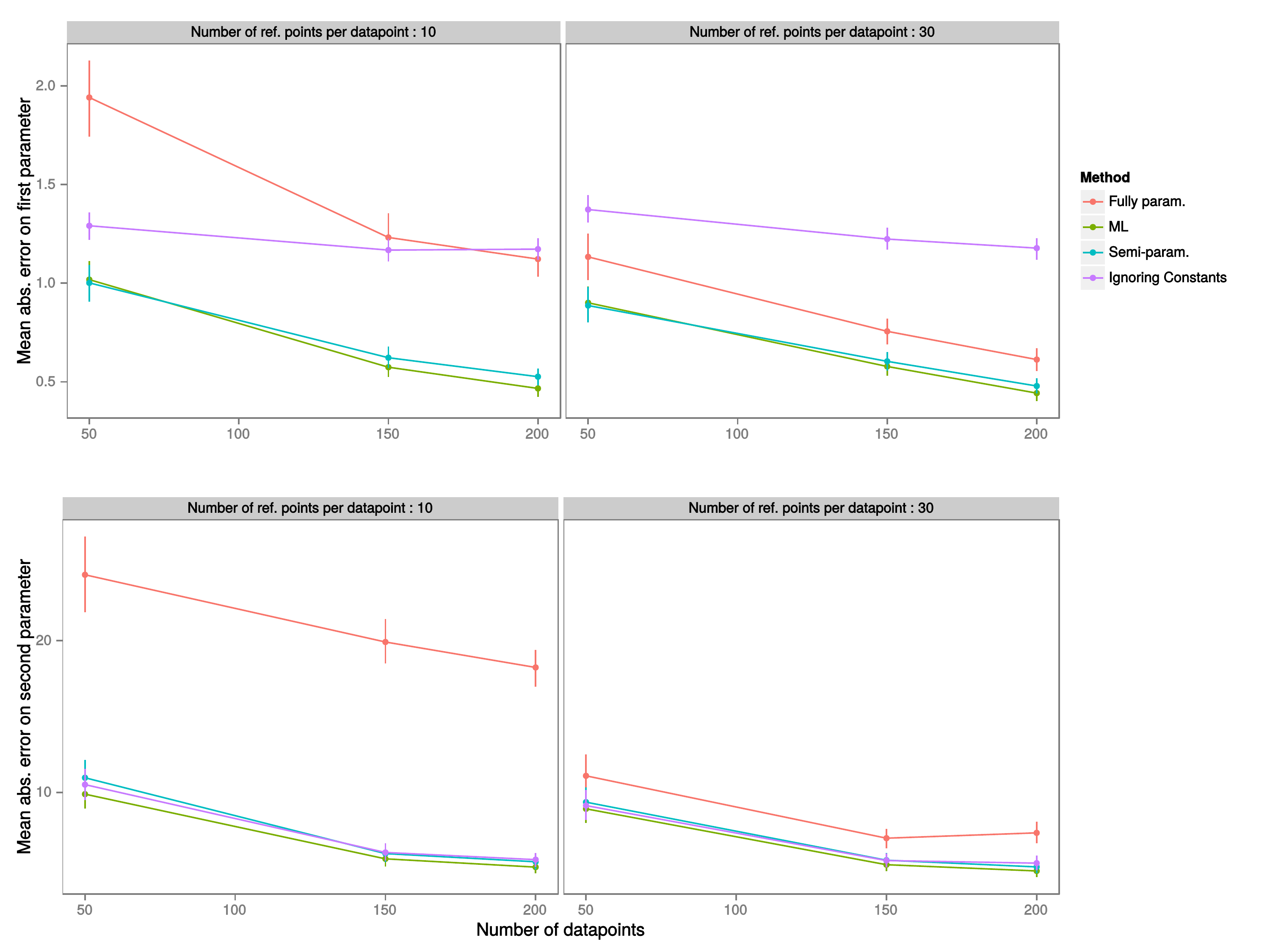}
\par\end{centering}

\protect\caption{Estimation errors of ML vs three variants of NCD for the one-dimensional
Markov chain. The variants are: fully parametric (one $\nu_{i}$ term
per datapoint), semi-parametric (normalisation constants are modelled
as a smooth function), ignoring constants (logistic regression with
a single intercept, as in \citealp{MnihTeh2012:FastSimpleAlgorithmTrainingNeuralProb}).\textbf{
}The semi-parametric estimate is almost as good as the ML estimate
across the board. The fully parametric estimate performs very poorly
when there are few reference points per datapoint (compare the red
line across the left and right panels). Finally, neglecting normalisation
constants leads to an non-convergent estimator of $\theta_{1}$, although
performance on $\theta_{2}$ is very good. See the companion document
for a more thorough discussion of this phenomenon.\label{fig:Results-estimation-toy}}
\end{figure}

\subsection{Spatial Markov chains for eye movement data\label{sub:Spatial-Markov-chains}}

\begin{figure}
\begin{centering}
\includegraphics[height=4cm]{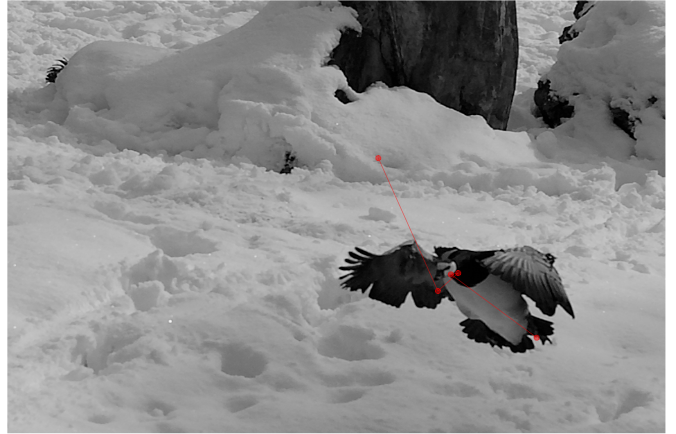}
\par\end{centering}

\protect\caption{A sequence of eye movements extracted from the dataset of \citet{Kienzle:CenterSurroundPatternsOptimalPredictors}.
Fixation locations are in red and successive locations are linked
by a straight line. \label{fig:Eye-movement-sequence}}
\end{figure}

A perennial problem in spatial statistics is to predict where certain
events are likely to take place (for example, cases of malaria in
a country) given past occurences and a set of spatial predictors (for
example, availability of mosquito nets). Point process models can
be used in such contexts, and one important class of applications
is to eye movement data \citep{Barthelme:ModelingFixationLocationsSpatialPointProc},
where the goal is to predict which locations people will look at in
a given visual stimulus (for example a photograph). Eye movements
are reliably drawn to certain features in a stimulus, but also exhibit
dependencies \citep{Engbert:SpatialStatsAndAttentionalDynSceneViewing},
and the most important of these is that we tend not to move our eyes
very much. If we are currently fixating on the bottom-left corner
of the screen, it will take a few steps for us to go look in the upper
right, even if there is something rather interesting there. 

The presence of dependencies motivates the introduction of models
of eye movements as \emph{spatial Markov chains. }Here we note $\mathbf{y}_{t}$
the fixation location at time $t$, and use a log-linear form for
the kernel:
\begin{equation}
p\left(\mathbf{y}_{t}|\mathbf{y}_{t-1}\right)\propto\exp\left\{ s\left(\mathbf{y}_{t}\right)+r\left(\mathbf{y}_{t},\mathbf{y}_{t-1}\right)\right\} \label{eq:spatial-markov-model}
\end{equation}
where $s(\mathbf{y}_{t})$ represents purely spatial factors, and
$r\left(\mathbf{y}_{t},\mathbf{y}_{t-1}\right)$ is an interaction
term that represents spatial dependencies. A well-known factor affecting
fixation locations is the centrality bias \citep{TatlerVincent:BehavBiasesEyeGuidance},
a preference for looking at central locations, and we take $s(\mathbf{y}_{t})$
to be a smooth function of $||\mathbf{y}_{t}||$ (the distance to
the center): $s(\by_{t})=s(\left\Vert \by_{t}\right\Vert )$. Potential
interactions between successive locations include a tendency not to
stray too far from the current location \citep{Engbert:SpatialStatsAndAttentionalDynSceneViewing},
and a tendency for making movements along the cardinal axes (vertical
and horizontal, \citealp{Foulsham:TurningTheWorldAround}). We therefore
further decompose $r\left(\mathbf{y}_{t},\mathbf{y}_{t-1}\right)$
into 
\begin{equation}
r\left(\mathbf{y}_{t},\mathbf{y}_{t-1}\right)=r_{\mathrm{dist}}\left(\left|\left|\mathbf{y}_{t}-\mathbf{y}_{t-1}\right|\right|\right)+r_{\mathrm{ang}}\left(\angle\left(\mathbf{y}_{t}-\mathbf{y}_{t-1}\right)\right)\label{eq:interaction-term}
\end{equation}
the sum of a distance and an angular component. We model the unknown
functions $s$, $r_{\mathrm{dist}}$ and $r_{\mathrm{ang}}$ non-parametrically,
using smoothing splines. The corresponding estimators are therefore
obtained by penalised likelihood maximisation, and the Poisson transform
extends straightforwardly to this case: replace the maximisation of
$\mathcal{L}(\bt)-\mathrm{pen(\bt)}$ by the maximisation of $\mathcal{M}(\bt,\chi)-\mathrm{pen(\bt)}$,
where $\bt=(s,r_{\mathrm{dist}},r_{\mathrm{ang}})$, and $\chi$ is
a non-parametric function used to estimate the normalising constant,
as explained in the previous section. 

We use the data of \citet{Kienzle:CenterSurroundPatternsOptimalPredictors},
who recorded eye movements while subjects where exploring a set of
photographs (Fig. \ref{fig:Eye-movement-sequence}). There are 14
subjects, each contributing between 600 and 2,000 datapoints. Thanks
to the techniques described above, the model described by \eqref{eq:spatial-markov-model}
can be turned into a logistic regression, and the R package \emph{mgcv
}\citep{Wood:GAMIntroductionWithR} can be used to estimate the different
components using smoothing splines. We used a uniform, IID reference
kernel $q(\mathbf{y}_{t}|\mathbf{y}_{t-1})=\left|\Omega\right|^{-1}$
to produce negative examples, with 20 times as many negative examples
as positive. Although the logistic approximation introduces Monte
Carlo variance, the estimates are very stable (see Appendix). We fit
separate functions for each subject to account for interindividual
variability. The results are shown on Fig. \ref{fig:results-eye-movements}.
We replicate known effects from the literature: central locations
dominate (although some subjects may display an off-center bias),
and dependencies include both a inhibitory effect of distance and
a preference for movements along cardinal orientations. 

Once the data have been put into a suitable format, model fitting
can be performed in one line of R code (see Appendix) and takes around
5 minutes on a normal desktop. The Poisson transform thus turns an
otherwise highly non-standard model into a convenient Generalised
Additive Model. 

\begin{figure}
\begin{centering}
\includegraphics[height=4cm]{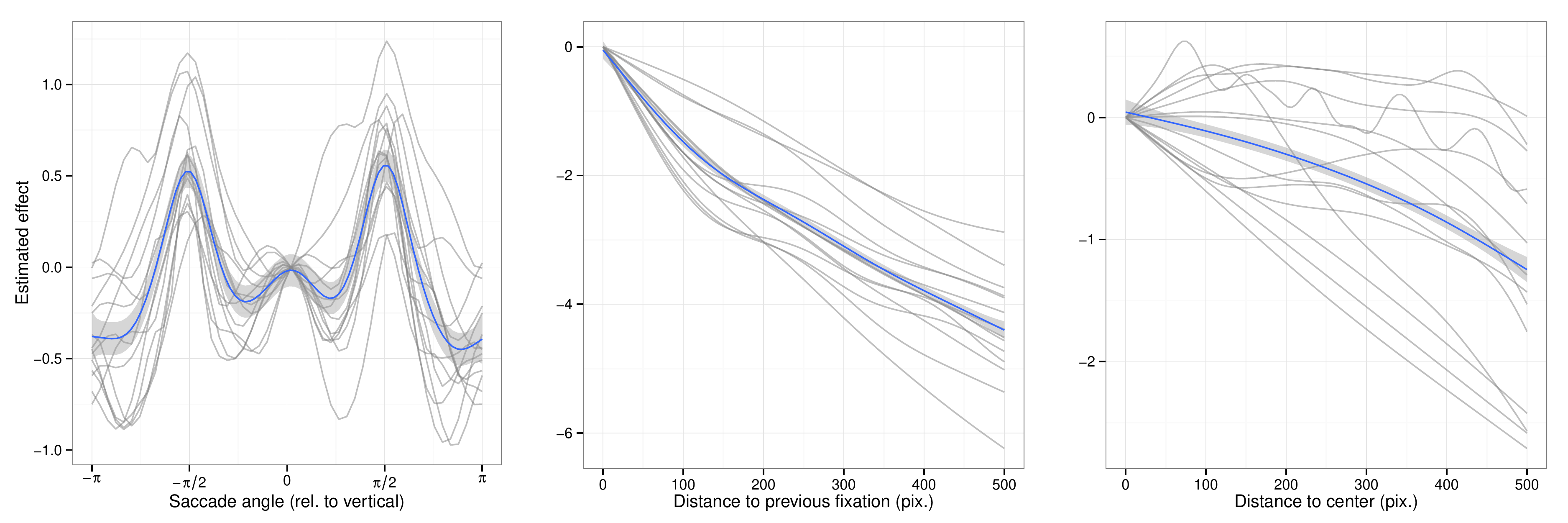}
\par\end{centering}

\protect\caption{Eye movement model. The smooth terms in eq. \ref{eq:spatial-markov-model}
and \ref{eq:interaction-term} are estimated using smoothing splines
by reducing the model to a non-parametric logistic regression. The
different panels display the estimated effects of saccade angle $(r_{ang})$,
distance to previous fixation ($r_{dist}$) and centrality bias $(s)$.
Individual subjects are in gray, and the group average is in blue.
\label{fig:results-eye-movements}}
\end{figure}

\section{Discussion}

The Poisson transform suggests a new way of thinking about inference
in unnormalised models: if we think of the data as coming from a point
process, the integration constant becomes just another parameter to
estimate. We have shown that the same idea extends to unnormalised
models in the sequential context and in the presence of covariates,
in which case parametric estimation may be turned into a semi-parametric
problem. Practical approximations of Poisson-transformed likelihoods
can be computed using Monte Carlo or using logistic likelihoods that
follow from a reinterpretation of noise-contrastive divergence. 

Part of the challenge in applying the Poisson transform to models
with high-dimensional covariates or dependencies on a high-dimensional
vector of past values will be in the design of appropriate kernels
for the non-parametric part, which corresponds to conditional normalisation
constants. The great advantage of the reduction to logistic regression
is that we will be able to leverage the existing literature on nonlinear
classification and dimensionality reduction, including recent developments
in hashing \citep{LiKoenig:TheoryApplicationsBbitMinwiseHashing}.
Inference in unnormalised models will probably always remain challenging,
but we believe the Poisson transform should alleviate some of the
difficulties.

\appendix

\section{Derivatives of Poisson-transformed likelihoods\label{sec:Derivatives-of-Poisson-transform}}

The first and second derivatives of $\L\left(\bt\right)$ and $\Mtn$
are needed in the proofs and we collect them here.

Derivatives of $\Lt$:

\begin{eqnarray*}
\mathcal{L}(\bt) & = & \sum_{i=1}^{n}f_{\bt}(\by_{i})-n\log\left(\int\mbox{exp}\left\{ f_{\bt}\left(s\right)\right\} \mbox{d}s\right):=\sum_{i=1}^{n}f_{\bt}(\by_{i})-n\phi\left(\bt\right)\\
\frac{\partial}{\partial\bt}\phi\left(\bt\right) & = & \int\frac{\partial}{\partial\bt}f_{\bt}(s)\mbox{exp}\left\{ f_{\bt}\left(s\right)-\phi\left(\bt\right)\right\} \mbox{d}s=E_{\bt}\left(\frac{\partial}{\partial\bt}f_{\bt}\right)\\
\frac{\partial^{2}}{\partial\bt^{2}}\phi\left(\bt\right) & = & E_{\bt}\left(\frac{\partial^{2}}{\partial\bt^{2}}f_{\bt}\right)+E_{\bt}\left(\left(\frac{\partial}{\partial\bt}f_{\bt}\right)\left(\frac{\partial}{\partial\bt}f_{\bt}\right)^{t}\right)-E_{\bt}\left(\frac{\partial}{\partial\bt}f_{\bt}\right)E_{\bt}\left(\frac{\partial}{\partial\bt}f_{\bt}\right)^{t}\\
\frac{\partial}{\partial\bt}\L & = & \sum_{i=1}^{n}\frac{\partial}{\partial\bt}f_{\bt}(\by_{i})-n\frac{d}{d\bt}\phi\left(\bt\right)\\
\frac{\partial^{2}}{\partial\bt^{2}}\L\left(\bt\right) & = & \sum\frac{\partial^{2}}{\partial\bt^{2}}f_{\bt}(\by_{i})-n\frac{d^{2}}{d\bt^{2}}\phi\left(\bt\right)
\end{eqnarray*}
where we have used $E_{\bt}$ as shorthand for the expectation with
respect to density $\mbox{exp}\left\{ f_{\bt}\left(s\right)-\phi\left(\bt\right)\right\} $. 

Derivatives of $\Mtn$:

\begin{eqnarray*}
\M\left(\bt,\nu\right) & = & \sum_{i=1}^{n}\left\{ f_{\bt}(\by_{i})+\nu\right\} -n\int\mbox{exp}\left\{ f_{\bt}\left(s\right)+\nu\right\} \mbox{d}s\\
\frac{\partial}{\partial\bt}\M\left(\bt,\nu\right) & = & \sum_{i=1}^{n}\frac{\partial}{\partial\bt}f_{\bt}(\by_{i})-nE_{\bt,\nu}\left(\frac{\partial}{\partial\bt}f_{\bt}\right)\\
\frac{\partial}{\partial\nu}\M\left(\bt,\nu\right) & = & n-n\int\mbox{exp}\left\{ f_{\bt}\left(s\right)+\nu\right\} \mbox{d}s\\
\frac{\partial^{2}}{\partial\bt^{2}}\M\left(\bt,\nu\right) & = & \sum\frac{\partial^{2}}{\partial\bt^{2}}f_{\bt}\left(\by_{i}\right)-n\left(E_{\bt,\nu}\left(\frac{\partial^{2}}{\partial^{2}\bt}f_{\bt}\right)+E_{\bt,\nu}\left(\left(\frac{\partial}{\partial\bt}f_{\bt}\right)\left(\frac{\partial}{\partial\bt}f_{\bt}\right)^{t}\right)\right)\\
\frac{\partial^{2}}{\partial\nu^{2}}\M\left(\bt,\nu\right) & = & -n\int\mbox{exp}\left\{ f_{\bt}\left(s\right)+\nu\right\} \mbox{d}s\\
\frac{\partial}{\partial\bt\partial\nu}\M\left(\bt,\nu\right) & = & -nE_{\bt,\nu}\left(\frac{\partial}{\partial\bt}f_{\bt}\right)
\end{eqnarray*}
where we have used $E_{\bt,\nu}$ as shorthand for the linear operator
$E_{\bt,\nu}(\varphi)=\int\varphi(s)\mbox{exp}\left\{ f_{\bt}\left(s\right)+\nu\right\} \,\mbox{d}s$
(which is not an expectation in general).

\section{Further properties of the Poisson transform\label{sec:Further-properties-of-Poisson-transform}}

\subsection{The Poisson transform preserves confidence intervals\label{sub:Preservation-confint}}

The usual method for obtaining confidence intervals for $\bt$ is
to invert the Hessian matrix of $\Lt$ at the mode, $\bts$: 

\[
\mathbf{C}_{\L}=\left(-\frac{d^{2}}{d^{2}\bt}\L\left|_{\bt=\bt^{\star}}\right.\right)^{-1}
\]
We can show that the same confidence intervals can be obtained from
$\Mtn$ at the joint mode, $\bts,\nu^{\star}$.

At the joint maximum, $\nu^{\star}$ normalises the intensity function,
and the Hessian of $\M$ equals:

\begin{eqnarray*}
H & = & \left[\begin{array}{cc}
\mathbf{H}_{aa} & \mathbf{H}_{ba}\\
\mathbf{H}_{ab} & \mathbf{H}_{bb}
\end{array}\right]=\left[\begin{array}{cc}
\frac{\partial^{2}}{\partial^{2}\bt}\M\left(\bt,\nu\right) & \frac{\partial}{\partial\bt\partial\nu}\M\left(\bt,\nu\right)\\
\frac{\partial}{\partial\nu\partial\bt}\M\left(\bt,\nu\right) & \frac{\partial^{2}}{\partial^{2}\nu}\M\left(\bt,\nu\right)
\end{array}\right]\\
 & = & \left[\begin{array}{cc}
\sum\frac{\partial^{2}}{\partial\bt^{2}}f_{\bt}\left(\by_{i}\right)-nE_{\bt}\left(\frac{\partial^{2}}{\partial^{2}\bt}f\right)-nE_{\bt}\left(\left(\frac{\partial}{\partial\bt}f_{\bt}\right)\left(\frac{\partial}{\partial\bt}f_{\bt}\right)^{t}\right) & -nE_{\bt}\left(\frac{\partial}{\partial\bt}f_{\bt}\right)^{t}\\
-nE_{\bt}\left(\frac{\partial}{\partial\bt}f_{\bt}\right) & -n
\end{array}\right]
\end{eqnarray*}
where again $E$ denotes the expectation with respect to density $\mbox{exp}\left\{ f_{\bt}\left(s\right)-\phi\left(\bt\right)\right\} $. 

Inverting $-H$ also yields confidence intervals. By the inversion
rule for block matrices, the approximate covariance for $\bt$ using
$\M\left(\bt,\nu\right)$ equals

\begin{eqnarray*}
\mathbf{C}_{\M}^{-1} & = & -\left(\mathbf{H}_{aa}-\mathbf{H}_{ba}\mathbf{H}_{bb}^{-1}\mathbf{H}_{ab}\right)\\
 & = & -\left(\mathbf{H}_{aa}+\frac{1}{n}n^{2}E_{\bt}\left(\frac{\partial}{\partial\bt}f_{\bt}\right)E_{\bt}\left(\frac{\partial}{\partial\bt}f_{\bt}\right)^{t}\right)\\
 & = & -\Bigg[\sum\frac{\partial^{2}}{\partial\bt^{2}}f_{\bt}\left(\by_{i}\right)-nE_{\bt}\left(\frac{\partial^{2}}{\partial^{2}\bt}f\right)\\
 &  & -nE_{\bt}\left(\left(\frac{\partial}{\partial\bt}f_{\bt}\right)\left(\frac{\partial}{\partial\bt}f_{\bt}\right)^{t}\right)+nE_{\bt}\left(\frac{\partial}{\partial\bt}f_{\bt}\right)E\left(\frac{\partial}{\partial\bt}f_{\bt}\right)^{t}\Bigg]\\
 & = & \mathbf{C}_{\L}^{-1}
\end{eqnarray*}

\subsection{Preservation of log-concavity in exponential families\label{sub:Preservation-log-concavity}}

In exponential families, the log-likelihood is concave, which facilitates
inference. The Poisson transform preserves this log-concavity.

In the natural parameterisation, exponential-family models are given
by:

\[
\Lt=\exp\left\{ \sum_{i=1}^{n}s(\mathbf{y}_{i})^{t}\bt-\phi\left(\bt\right)\right\} 
\]
with $s(\mathbf{y})$ a vector of sufficient statistics. The second
derivative of $\Lt$ simplifies to:
\begin{eqnarray*}
-\frac{1}{n}\frac{\partial}{\partial^{2}\bt}\Lt & = & \int s(\mathbf{y})s(\mathbf{y})^{t}\exp\left(s(\mathbf{y})^{t}\bt-\phi\left(\bt\right)\right)\\
 & = & E_{\bt}\left\{ s(\mathbf{y})s(\mathbf{y})^{t}\right\} 
\end{eqnarray*}
a p.s.d. matrix, which establishes concavity. 

The second derivatives of $\Mtn$ (Section \ref{sec:Derivatives-of-Poisson-transform})
also simplify
\begin{eqnarray*}
-\frac{1}{n}\frac{\partial^{2}}{\partial\bt^{2}}\Mtn & = & \exp\left\{ \nu-\nu^{\star}\left(\bt\right)\right\} \int s(\mathbf{y})s(\mathbf{y})^{t}\exp\left(s(\mathbf{y})^{t}\bt-\phi\left(\bt\right)\right)\\
-\frac{1}{n}\frac{\partial^{2}}{\partial\nu\partial\bt}\Mtn & = & \exp\left\{ \nu-\nu^{\star}\left(\bt\right)\right\} \int s(\mathbf{y})\exp\left(s(\mathbf{y})^{t}\bt-\phi\left(\bt\right)\right)\\
-\frac{1}{n}\frac{\partial^{2}}{\partial\nu{}^{2}}\Mtn & = & \exp\left\{ \nu-\nu^{\star}\left(\bt\right)\right\} 
\end{eqnarray*}
so that the full Hessian $\mathbf{H}$ can be written in block-form
as:

\[
-\frac{1}{n}\exp\left\{ \nu^{\star}\left(\bt\right)-\nu\right\} \mathbf{H}=\left[\begin{array}{cc}
E_{\bt}\left\{ s\left(\mathbf{y}\right)s\left(\mathbf{y}\right)^{t}\right\}  & E\left(s\left(\mathbf{y}\right)\right)\\
E_{\bt}\left\{ s\left(\mathbf{y}\right)^{t}\right\}  & 1
\end{array}\right]=\mathbf{A}
\]
and $\mathbf{H}$ is n.s.d if and only if for all $\mathbf{x},c$
such that $(\bm{x},c)\neq\bm{0}$:

\[
\left[\begin{array}{cc}
\mathbf{x}^{t} & c\end{array}\right]\mathbf{A}\left[\begin{array}{c}
\mathbf{x}\\
c
\end{array}\right]>0
\]
which the following establishes: 

\begin{align*}
 & \left[\begin{array}{cc}
\mathbf{x}^{t} & c\end{array}\right]\left[\begin{array}{cc}
E_{\bt}\left\{ s\left(\mathbf{y}\right)s\left(\mathbf{y}\right)^{t}\right\}  & E\left\{ s\left(\mathbf{y}\right)\right\} \\
E_{\bt}\left\{ s\left(\mathbf{y}\right)^{t}\right\}  & 1
\end{array}\right]\left[\begin{array}{c}
\mathbf{x}\\
c
\end{array}\right]\\
= & \mbox{\ensuremath{\left[\begin{array}{cc}
 \mathbf{x}^{t}  &  c\end{array}\right]}}\left[\begin{array}{c}
E_{\bt}\left\{ s\left(\mathbf{y}\right)s\left(\mathbf{y}\right)^{t}\right\} \mathbf{x}+cE_{\bt}\left\{ s\left(\mathbf{y}\right)\right\} \\
E_{\bt}\left\{ s\left(\mathbf{y}\right)^{t}\right\} \mathbf{x}+c
\end{array}\right]\\
= & E_{\bt}\left\{ \mathbf{x}^{t}s\left(\mathbf{y}\right)s\left(\mathbf{y}\right)^{t}\mathbf{x}\right\} +2E_{\bt}\left\{ \mathbf{x}^{t}s\left(\mathbf{y}\right)\right\} c+c^{2}\\
= & E_{\bt}\left[\left(s\left(\mathbf{y}\right)^{t}\mathbf{x}+c\right)^{2}\right]>0
\end{align*}
assuming $E_{\bt}\left\{ s(\by)s(\by)^{t}\right\} $ is p.s.d. for
all $\bt$.

\subsection{Noise-constrative divergence approximates the Poisson transform (Theorem
\ref{thm:logistic-reg-tends-to-IPP})}

We have assumed that 

\[
f_{\bt}(\by)-\log q(\by)\leq C(\bt)
\]
for a certain constant $C(\bt)$ that may depend on $\bt$, and all
$\by\in\Omega$. We rewrite the log-odds ratio as $h(\by)-\log(m)$
where 
\[
h(\by):=f_{\bt}(\by)+\nu-\log q(\by)+\log(n)
\]
does not depend on $m$; note $h(\by)\leq\bar{h}:=C(\bt)+\nu+\log(n)$.
One has: 
\begin{align*}
\mathcal{R}^{m}(\bt,\nu)+\log(m/n)= & \sum_{i=1}^{n}\log\left[\frac{m\exp\left\{ f_{\bt}(\by_{i})+\nu\right\} }{n\exp\left\{ f_{\bt}(\by_{i})+\nu\right\} +mq(\by_{i})}\right]\\
 & +\sum_{j=1}^{m}\log\left[\frac{mq(\bm{r}_{j})}{n\exp\left\{ f_{\bt}(\bm{r}_{j})+\nu\right\} +mq(\bm{r}_{j})}\right]
\end{align*}
where the first term trivially converges (as $m\rightarrow+\infty$)
to 
\[
\sum_{i=1}^{n}\left\{ f_{\bt}(\by_{i})+\nu-\log q(\by_{i})\right\} .
\]
Regarding the second term, one has:
\[
\log\left[\frac{mq(\bm{r}_{j})}{n\exp\left\{ f_{\bt}(\bm{r}_{j})+\nu\right\} +mq(\bm{r}_{j})}\right]=\log\left[1-\frac{1}{1+m\exp\left\{ -h(\bm{r}_{j})\right\} }\right]
\]
where
\[
0\leq\frac{1}{1+m\exp\left\{ -h(\bm{r}_{j})\right\} }\leq\frac{1}{m}\exp(\bar{h}).
\]
Since $\left|\log(1-x)+x\right|\leq x^{2}$ for $x\in[0,1/2]$, we
have, for $m$ large enough, that
\begin{equation}
\left|\log\left[\frac{mq(\bm{r}_{j})}{n\exp\left\{ f_{\bt}(\bm{r}_{j})+\nu\right\} +mq(\bm{r}_{j})}\right]+\frac{1}{1+m\exp\left\{ -h(\bm{r}_{j})\right\} }\right|\leq\frac{\exp(2\bar{h})}{m^{2}}\label{eq:second_bound}
\end{equation}
and 
\[
\left|\frac{1}{1+m\exp\left\{ -h(\bm{r}_{j})\right\} }-\frac{1}{m}\exp\left\{ h(\bm{r}_{j})\right\} \right|\leq\frac{\exp(2\bar{h})}{m^{2}}
\]
and since, by the law of large numbers, 
\begin{equation}
\frac{1}{m}\sum_{j=1}^{m}\exp\left\{ h(\bm{r}_{i})\right\} \rightarrow\mathbb{E}_{q}[\exp\left\{ h(\bm{r}_{i})\right\} ]=n\int\exp\left\{ f_{\bt}(\by)+\nu\right\} \, d\by<+\infty\label{eq:LLN}
\end{equation}
almost surely as $m\rightarrow+\infty$, one also has:
\[
\sum_{j=1}^{m}\log\left[\frac{mq(\by_{i})}{n\exp\left\{ f_{\bt}(\by_{i})+\eta\right\} +mq(\by_{i})}\right]\rightarrow-n\int\exp\left\{ f_{\bt}(\by)+\nu\right\} \, d\by
\]
almost surely, since the difference between the two sums is bounded
deterministically by $\exp(2\bar{h})/m$.

\subsection{Uniform convergence of the noise-constrative divergence (Theorem
\ref{thm:MLEconverges})}

We first prove two intermediate results. 
\begin{lem}
\label{lem:bound_nu}Assuming that $\left|f_{\bt}(\by)-\log q(\by)\right|\le C$
for all $\by\in\Omega$, then there exists a bounded interval $I$
such that, for any $\bt$, the maximum of both functions $\nu\rightarrow\M(\bt,\nu)$
and $\nu\rightarrow\mathcal{R}^{m}(\bt,\nu)$ is attained in $I$. \end{lem}
\begin{proof}
Let $\bt$ some fixed value. $\M(\bt,\nu)$ is maximised at $\nu^{\star}(\bt)=-\log\int_{\Omega}\exp\left\{ f_{\bt}(\by)\right\} \, d\by\in\left[-C,C\right]$,
since $e^{-C}q\leq f_{\bt}\leq e^{C}q$. For $\mathcal{R}^{m}(\bt,\nu)$,
using again $e^{-C}q\leq f_{\bt}\leq e^{C}q$, one sees that $l(\nu)\leq\mathcal{R}^{m}(\bt,\nu)\leq u(\nu)$,
where $l$ and $u$ are functions of $\nu$ that diverges at $-\infty$
for both $\nu\rightarrow+\infty$ and $\nu\rightarrow-\infty$; i.e.
\begin{align*}
\mathcal{R}^{m}(\bt,\nu)+\log(m/n)\leq u(\nu):= & \sum_{i=1}^{n}\log\left[\frac{m\exp\left\{ C+\nu\right\} }{n\exp(-C+\nu)+m}\right]\\
 & +\sum_{j=1}^{m}\log\left[\frac{m}{n\exp\left\{ -C+\nu\right\} +m}\right]
\end{align*}
and the lower bound $l(\nu)$ has a similar expression. Thus one may
construct an interval $J$ such that the maximum of function $\nu\rightarrow\mathcal{R}^{m}(\bt,\nu)$
is attained in $J$ for all $\bt$ (e.g. take $J$ such that for $\nu\in J^{c}$,
$u(\nu)\leq M_{l}/2$, $l(\nu)\leq M_{l}/2$, with $M_{l}=\sup_{\nu}l$)
. To conclude, take $I=J\cup[-C,C]$. 
\end{proof}
We now establish uniform convergence, but, in light of the previous
result, we restrict $\nu$ to the interval $I$ defined in Lemma \ref{lem:bound_nu}. 
\begin{lem}
Under the Assumptions that (i) $\Theta$ is bounded, that (ii) $\left|f_{\bt}(\by)-\log q(\by)\right|\le C$
for all $\by\in\Omega$, that (iii) $\left|f_{\bt}(\by)-f_{\bt'}(\by)\right|\leq\kappa(\by)\left\Vert \bt-\bt'\right\Vert $
with $\kappa$ such that $\E_{q}[\kappa]<\infty$, one has, for fixed
$\mathcal{S}=\left\{ \by_{1},\ldots,\by_{n}\right\} $:
\begin{equation}
\sup_{\left(\bt,\nu\right)\in\Theta\times I}\left|\mathcal{R}^{m}(\bt,\nu)+\log(m/n)+\sum_{i=1}^{n}\log q(\bm{y}_{i})-\M(\bt,\nu)\right|\rightarrow0\label{eq:uniform_conv}
\end{equation}
almost surely, relative to the randomness induced by $\mathcal{R}=\left\{ \bm{r}_{1},\ldots,\bm{r}_{m}\right\} .$ \end{lem}
\begin{proof}
Recall that the absolute difference above was bounded by the sum of
three terms in the previous Appendix. The first term was 
\[
\sum_{i=1}^{n}\left[\log\left[\frac{m\exp\left\{ f_{\bt}(\by_{i})+\nu\right\} }{n\exp\left\{ f_{\bt}(\by_{i})+\nu\right\} +mq(\by_{i})}\right]-\left\{ f_{\bt}(\by_{i})+\nu-\log q(\by_{i})\right\} \right]
\]
which clearly converges deterministically to $0$ as $m\rightarrow+\infty.$
In addition, this convergence is uniform with respect to $\left(\bt,\nu\right)\in\Theta\times I$,
since $\left|\log x-\log y\right|\leq c\left|x-y\right|$ for $x,y\geq1/c$,
and here, by Assumption (ii), 
\[
x:=\frac{m\exp\left\{ f_{\bt}(\by_{i})+\nu\right\} }{n\exp\left\{ f_{\bt}(\by_{i})+\nu\right\} +mq(\by_{i})}\geq\frac{m\exp\left\{ -C+\nu\right\} }{n\exp\left\{ C+\nu\right\} +m}\geq\exp\left\{ -C+\nu\right\} 
\]
and $y=\exp\left\{ f_{\bt}(\by_{i})+\nu-\log q(\by_{i})\right\} \geq\exp\left\{ -C+\nu\right\} $,
so both $x$ and $y$ are lower bounded since $\nu\in I$. Similarly
$(x-y)$ is bounded by $C'/m$, where $C'$ is some constant independent
of $\bt$. 

The second term, see \eqref{eq:second_bound}, was bounded by $\exp(2\bar{h})/m^{2}$,
where $\bar{h}$, an upper bound of $h$, may now be replaced by a
constant, since $h(\by):=f_{\bt}(\by)+\nu-\log q(\by)+\log(n)\leq C+\nu+\log(n)$
and $\nu\in I$, again by Assumption (ii). 

The third term is related to the law of large numbers \eqref{eq:LLN}
for random variable $H_{(\bt,\eta)}(\bm{r}_{i}):=\exp\left\{ h(\bm{r}_{i})\right\} $,
which depended implicitly on $\left(\bt,\eta\right)$: 
\[
H_{(\bt,\eta)}(\bm{r}_{i})=\frac{n\exp\left\{ f_{\bt}(\bm{r}_{i})+\nu\right\} }{q(\bm{r}_{i})}.
\]
To obtain (almost surely) uniform convergence, we use the generalised
version of the Glivenko-Cantelli theorem; e.g. Theorem 19.4 p.270
in \citet{VanderVaart2007}. From Example 19.7 of the same book, one
sees that a sufficient condition in our case is that $\Theta$ is
bounded (Assumption (i)), and that 
\[
\left|H_{(\bt,\eta)}(\bm{r})-H_{(\bt',\eta')}(\bm{r})\right|\leq m(\bm{r})\left\Vert \bxi-\bxi'\right\Vert 
\]
for $\bxi=(\bt,\eta)$, $\bxi'=(\bt',\eta')$, and $m$ a function
such that $\E_{q}[m]<\infty$. But
\begin{align*}
\left|H_{(\bt,\eta)}(\bm{r})-H_{(\bt',\eta')}(\bm{r})\right| & =\frac{n\exp\left\{ f_{\bt}(\bm{r})+\nu\right\} }{q(\bm{r})}\left|1-\exp\left\{ f_{\bt}(\bm{r})+\nu-f_{\bt'}(\bm{r})-\nu'\right\} \right|\\
 & \leq ne^{C+\nu}\left|1-\exp\left\{ f_{\bt}(\bm{r})+\nu-f_{\bt'}(\bm{r})-\nu'\right\} \right|\\
 & \leq C'\left\{ \kappa(\bm{r})\left\Vert \bt-\bt'\right\Vert +\left|\nu-\nu'\right|\right\} \\
 & \leq C'\left\{ \kappa(\bm{r})+1\right\} \left\Vert \bxi-\bxi'\right\Vert 
\end{align*}
by Assumption (ii), and for some constant $C'$ independent of $\bt$,
since $\left|1-e^{x}\right|\leq Kx$ for $x,y$ in a bounded set.
One may conclude, since, by Assumption (ii), $\E_{q}[\kappa]<\infty$.
\end{proof}
We are now able to prove Theorem \ref{thm:MLEconverges}. Again, let
$\bxi=(\bt,\nu)$, and rewrite any function of $(\bt,\nu)$ as a function
of $\bxi$, i.e. $\M(\bxi)$, $\mathcal{R}^{m}(\bxi)$. By e.g. Theorem
5.7 p.45 of \citet{VanderVaart2007}, the uniform convergence \ref{eq:uniform_conv}
implies that that the maximiser $\hat{\bxi}^{m}$ of $\mathcal{R}^{m}(\bt,\nu)$
converges to the maximiser $\hat{\bxi}$ of $\M(\bt,\nu)$, provided
that (a) the maximisation is with respect to $\left(\bt,\nu\right)\in\Theta\times I$;
and (b) that $\sup_{d(\bxi,\hat{\bxi})\geq\epsilon}\mathcal{M}(\bxi)<\mathcal{M}(\hat{\bxi})$.
However, by Lemma \ref{lem:bound_nu} one sees that in (a) the same
estimators would be obtained by maximising instead with respect to
$\left(\bt,\nu\right)\in\Theta\times\R$, and (b) is a direct consequence
of Assumption (iv) of the theorem, if one takes for  $d(\bxi,\hat{\bxi})$
the supremum norm of $\bxi-\hat{\bxi}$.

\section{Additional information on the application}

In our application we fit a spatial Markov chain model using logistic
regression. Since the procedure involves the generation of a random
set of reference points, we incur some Monte Carlo error in the estimates.
Estimating the magnitude of the Monte Carlo error is just a matter
of running the procedure several times to look at variability in the
estimates. We did so over 5 repetitions and report the results in
Fig. \ref{fig:replicates-fit}. For each repetition we plot the estimated
smooth effect of saccade angle $r_{\mathrm{ang}}$, along with a 95\%
confidence band. Since smoothing splines are used, smoothing hyperparameters
had to be inferred from the data (using REML, \citealp{Wood:FastStableMaxLikEstSemiParamGLMs}),
and the reported confidence band is \emph{conditional }on the estimated
value of the smoothing hyperparameters. The fits and confidence bands
are extremely stable over independent repetitions. The $R$ command
we used was: 
\begin{quotation}
\textsf{gam(class \textasciitilde{} s(delta,k=10)+s(dcenter,k=40)+s(fxc.prev,fyc.prev,k=40)}

\textsf{+s(angle,bs=\textquotedbl{}cc\textquotedbl{},k=20),data=data,family=binomial,method=''REML'')}
\end{quotation}
\begin{figure}
\begin{centering}
\includegraphics[height=10cm]{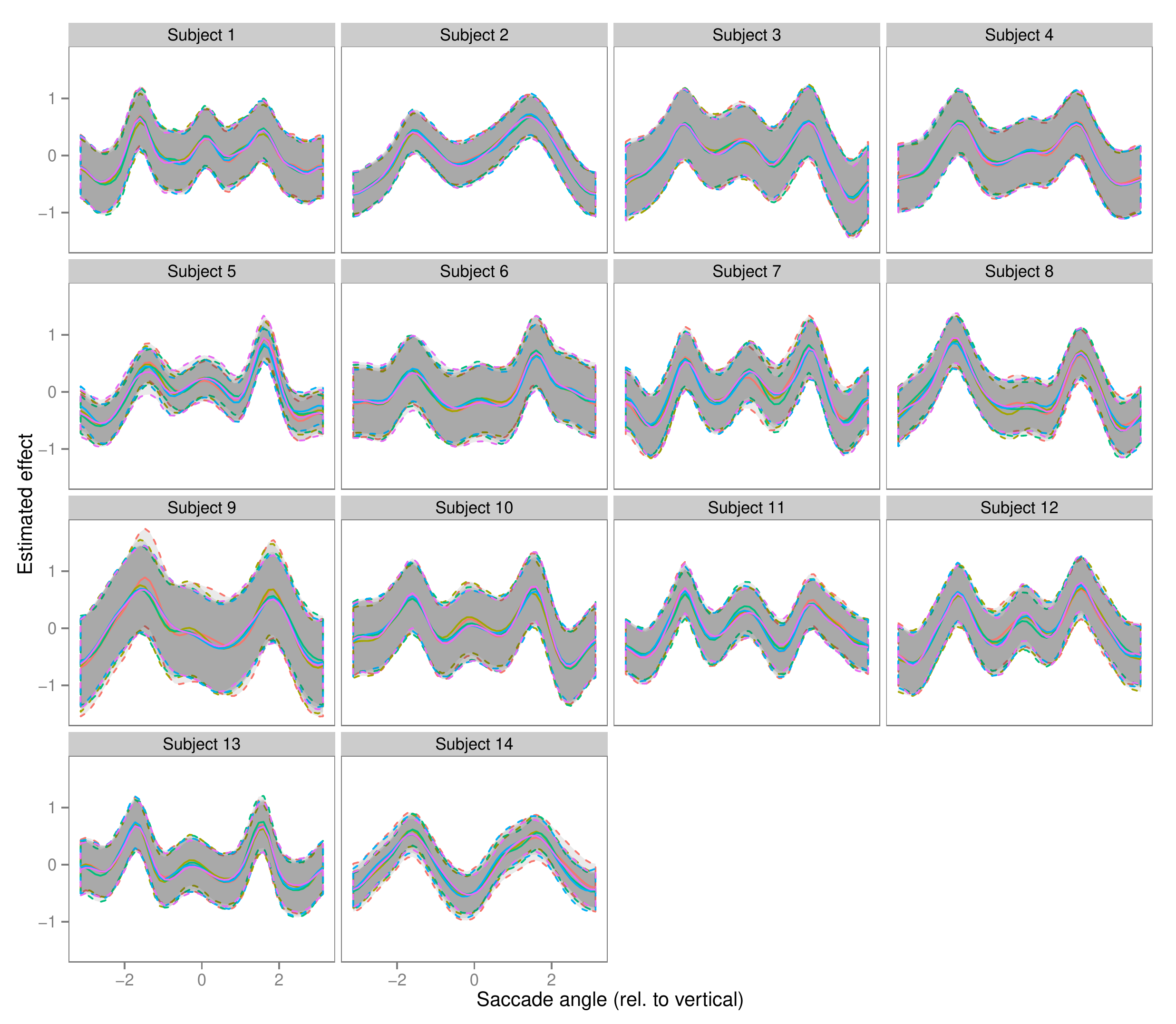}
\par\end{centering}

\protect\caption{Eye movement model: 5 independent replications of the estimates under
different sets of random reference points. We show here the estimated
effect of saccade angle with an associated 95\% pointwise confidence
interval. The 5 replicates are in different colours and overlap each
other almost completely, showing that 20 reference points per true
datapoint are more than enough to produce stable estimates. \label{fig:replicates-fit}}
\end{figure}

\bibliographystyle{apalike}
\bibliography{ref}

\end{document}